\documentclass[acmsmall,authorversion,nonacm,screen]{acmart}

\usepackage{amsfonts}
\usepackage{algorithm2e}
\RestyleAlgo{ruled}
\LinesNumbered
\usepackage{bm}
\usepackage{hyperref}
\usepackage{lscape}
\usepackage{multirow}
\usepackage{physics}
\usepackage[figuresright]{rotating}
\usepackage{soul}
\usepackage{subcaption}

\def\cspace{\mathbb{C}^n}
\def\posint{\mathbb{Z}^+}
\def\sysspace{\mathcal{Z}}
\def\posint{\mathbb{N}_{\geq 0}}
\newcommand{\init}[1]{#1_0}
\newcommand{\unsafe}[1]{#1_u}

\def\timesys{(\sysspace, \init{\sysspace}, F, f)}
\def\iu{i}

\def\ie{\emph{i.e.}, }
\def\eg{\emph{e.g.}, }

\AtEndPreamble{%
\theoremstyle{acmdefinition}
\newtheorem{remark}[theorem]{Remark}}

\begin{document}
\title{Verification of Quantum Circuits through Discrete-Time Barrier Certificates}
\author{Marco Lewis}
\email{m.j.lewis2@newcastle.ac.uk}
\affiliation{
    \department{School of Computing}
    \institution{Newcastle University}
    \streetaddress{1 Science Square}
    \city{Newcastle upon Tyne}
    \country{UK}
    \postcode{NE4 5TG}
}
\author{Sadegh Soudjani}
\affiliation{
    \institution{Max Planck Institute for Software Systems}
    \streetaddress{Paul Ehrlich Str. 26}
    \city{Kaiserslautern}
    \country{Germany}
    \postcode{67663}
}
\author{Paolo Zuliani}
\affiliation{
    \institution{Università di Roma ``La Sapienza''}
    \department{Dipartimento di Informatica}
    \streetaddress{via Salaria 133}
    \city{Rome}
    \country{Italy}
    \postcode{00198}
}



\begin{abstract}
    Current methods for verifying quantum computers are predominately based on interactive or automatic theorem provers.
    Considering that quantum computers are dynamical in nature, this paper employs and extends the concepts from the verification of dynamical systems to verify properties of quantum circuits.
    Our main contribution is to propose k-inductive barrier certificates over complex variables and show how to compute them using Hermitian Sum of Squares  optimization.
    We apply this new technique to verify properties of different quantum circuits.
\end{abstract}

\maketitle

\section{Introduction}
\label{sec:intro}
Various techniques for verifying quantum programs and circuits have been explored in the past two decades.
Many of the tools currently available involve using quite manual verification techniques such as theorem provers~\cite{SQIR,QHLProver,QHLProver-AFP}.
Other approaches have investigated using automated techniques such as Satisfiability Modulo Theory (SMT) solvers~\cite{QBricks,symQV}, automata-based methods~\cite{AutoQ}, and abstract interpretation~\cite{QuantumAbstractInterp}.
Whilst theorem provers have the benefit of scalability, automated approaches are less manual, only requiring the circuit or program to be proved and the specification of the desired behavior.

In recent work \cite{Lewis23}, we proposed a new automated approach for analyzing the behavior of quantum systems through the usage of barrier certificates~\cite{Dai17}.
Barrier certificates are a technique used to solve safety problems, where we have a dynamical system (\ie we know its dynamics $\dv{x}{t}$) that we want to avoid entering an unsafe region.
The main issue faced in \cite{Lewis23} was how to adapt barrier certificates to the complex domain for continuous-time quantum systems.

There is an argument that one could simply transform a complex system into a real system and analyze the system that way.
This leads to two major issues.
Firstly, the number of variables for the system doubles by taking this into consideration potentially making it take longer to solve.
Secondly, the system or the meaning behind the barrier certificate generated becomes harder to interpret.
Additionally, many barrier certificates generation techniques are based on Sum of Squares (SOS) optimization, which makes use of semidefinite program solvers in order to find a solution.
Gilbert and Josz demonstrate in \cite{Gilbert17} that semidefinite programs can benefit from a $\times 2$ to a $\times 4$ speed-up when using complex variables over converting the variables from complex to real.
Therefore, it is worth exploring the adaptation of dynamical systems to complex variables.

Adapting problems from the real domain into the complex domain has been investigated for few safety problems.
Fang and Sun~\cite{Fang2013} have developed a technique for analyzing the stability of complex dynamical systems.
In \cite{Lewis23}, we adapted the technique developed by Fang and Sun to develop barrier certificates for complex variables.
We extend upon our previous work in the following directions:
\begin{enumerate}
    \item we investigate the safety of {\em discrete-time} dynamical systems using complex variables, adapting barrier certificate techniques from \cite{Prajna04} and \cite{Anand21};
    \item we show how Sum of Squares (SOS) optimization is generalized to the complex domain using Hermitian Sum of Squares (HSOS), giving a more general approach to barrier certificate generation than in our previous work; and
    \item we investigate the usage of our adapted techniques for verifying quantum circuits and apply them to multiple case studies.
\end{enumerate}
By showing how we can adapt SOS optimization for complex numbers, we provide a method that allows one to easily adapt any discrete-time dynamical system and barrier certificate definition from the real domain into the complex domain.

\paragraph{Related Work}
In the last few years, there have been substantial efforts made to verify properties of quantum computers and programs.
Various tools have been developed using a variety of different formal methods.
Such methods include theorem provers~\cite{CoqQ,SQIR,QHLProver}; abstract interpretation~\cite{QuantumAbstractInterp}; static analysis~\cite{LintQ}; model checking techniques~\cite{QReach, Guan24}; relational reasoning~\cite{Barthe19, Yan24}; techniques using SMT solvers~\cite{Giallar,QBricks}; and automata-based techniques~\cite{AutoQ}.
For more on formal verification of quantum computers, we refer to surveys on the topic~\cite{Chareton23, Lewis23survey}.

Another approach is to treat a quantum program or computer as a dynamical system and transform problems into reachability or safety verification problems.
There are several ways to solve the safety verification problem for dynamical systems.
One approach is to use abstract interpretation to give an abstraction of the system's evolution~\cite{Cousot77, Cousot2001}.
This approach has already been investigated for verification of quantum programs in \cite{QuantumAbstractInterp}.
Another approach is to use forward or backward reachability~\cite{SS2014precise, SS2015quantitative, Mitchell07}, evolving the system from the initial or unsafe region until the system is shown to be safe or unsafe.
Barrier certificates are another method for solving the problem of safety~\cite{Dai17}.
The benefit of barrier certificates in comparison to the other two approaches is that barrier certificates do not need to evolve the system (\ie solve its dynamics) to prove safety.
In addition, barrier certificates can prove safety over unbounded time horizons, which is often not possible with other techniques.

There are other barrier certificates that exist depending on the system used and the nature of the safety problem.
Whilst in this article we are investigating discrete-time dynamical systems~\cite{Agrawal17}, there are techniques for handling continuous time and hybrid systems~\cite{Prajna04}.
Recent research directions have explored barrier certificates for stochastic systems~\cite{TAKANO18}, where the system evolves with noise. The use of barrier certificates for specifications beyond safety is also studied (e.g., \cite{jagtap2020formal,kordabad2024control}), and data-driven techniques for computing barrier certificates are also developed \cite{schon2024data,salamati2024data}.

\section{Background}
\label{sec:background}
We begin by introducing notation and covering some common theory in quantum computing and safety using barrier certificates.
A deeper background to quantum computing can be found in Nielsen and Chuang's volume~\cite{nielsen_chuang_2010}.

\subsection{Notation}
Throughout we write:
\begin{itemize}
    \item the imaginary unit, $\iu = \sqrt{-1}$ (we do \emph{not} use $\iu$ as an iterator or variable);
    \item for a complex number $z = a + b\iu$, its complex conjugate is $\overline{z} = a - b\iu$;
    \item for an $n \times m$ matrix, $z$, we write
    \begin{itemize}
        \item $(z)_{jk}$ as an element of $z$ where $0 \leq j \leq n-1, 0 \leq k \leq m-1$ (for vectors, $m=1$, we may simply write $(z)_j$),
        \item $\overline{z}$ as the conjugate of $z$ ($(\overline{z})_{jk} = \overline{(z)_{jk}}$),
        \item $z^\intercal$ as the transpose of $z$ ($(z^\intercal)_{jk} = (z)_{kj}$),
        \item and $z^\dagger = \overline{z}^\intercal$ as the conjugate transpose of $z$;
    \end{itemize}
    \item $I_n$ for the $n \times n$ identity operation.
\end{itemize}

Additionally, we make use of sub-equations, \eg Equation \eqref{eq:gen:init}, \eqref{eq:gen:unsafe} and \eqref{eq:gen:barrier}.
We may refer to these equations collectively as \eqref{eq:gen}.

\subsection{Quantum Circuits}
\label{sec:background:quantum}
For an $n$-qubit system, a quantum state can be described by a complex vector $z \in \mathbb{C}^{2^n}$ such that $\sum_j |(z)_j|^2 = 1$.
A quantum state is commonly written as $\ket{\phi} = \sum_{j=0}^{2^n - 1} (z)_j \ket{j}$, where $\{ \ket{j}: j=0,\ldots , 2^n-1\}$ is the {\em diagonal} (orthonormal) basis.
Evolution of a system occurs through unitary operations; which are $2^n \times 2^n$ complex matrices, $U$, such that $U U^\dagger = U^\dagger U = I_{2^n}$.
We write the evolution of $\ket{\phi}$ according to $U$ as $U\ket{\phi}$.

Some common unitary operations include, for example, the Hadamard $(H)$, $Z$, NOT $(X)$, and  Controlled-NOT ($CNOT$) gates:
\begin{align}
    H = \frac{1}{\sqrt{2}}
    \begin{pmatrix}
        1 & 1 \\
        1 & -1
    \end{pmatrix},
    \quad
    Z = 
    \begin{pmatrix}
        1 & 0 \\
        0 & -1
    \end{pmatrix},
    \quad
    X =
    \begin{pmatrix}
        0 & 1 \\
        1 & 0
    \end{pmatrix},
    \quad
    CNOT =
    \begin{pmatrix}
        1 & 0 & 0 & 0 \\
        0 & 1 & 0 & 0 \\
        0 & 0 & 0 & 1 \\
        0 & 0 & 1 & 0
    \end{pmatrix}
    \label{eq:qops}
\end{align}
which behave as 
\begin{align*}
    &H \ket{j} = \frac{1}{\sqrt{2}}(\ket{0}+(-1)^j\ket{1}) \quad\quad
    &Z \ket{j} = (-1)^j \ket{j} \\
    &X \ket{j} = \ket{\neg j} \quad\quad 
    &CNOT \ket{j,k} = \ket{j, j\oplus k} 
\end{align*}
where $j,k\in\{0,1\}$ and $\oplus$ is the XOR Boolean operation.

\begin{remark}
    We use $U$ to represent the quantum operation or the corresponding matrix.
    \textit{I.e.}, we may write the evolution of a quantum state as $U \ket{\phi}$ or $U z$ (using the standard dot product).
\end{remark}

Measurement of a quantum state is another way to evolve a quantum system.
We are only interested in quantum circuits without measurement.
However, for our purposes, we only need to be concerned that a basis state, $\ket{j}$, of a quantum state, $\ket{\phi}$, is measured with a probability based on its amplitude: $P(\ket{j} \text{ measured from } \ket{\phi}) = \abs{(z)_j}^2$.
See Nielsen and Chuang's volume~\cite{nielsen_chuang_2010} for details on measurement.

A quantum circuit diagram consists of wires that each correspond to a qubit, and gates that perform unitary operations on qubits.
The quantum operations given in Equation~\eqref{eq:qops} are represented by the respective gates in Figure~\ref{fig:qopgates}.
A diagrammatic example can be found in Figure~\ref{fig:qcirc:ex}, where if $\ket{\phi}$ is the initial quantum state then the resulting quantum state is $U_{k-1} U_{k-2} \dots U_0 \ket{\phi}$.

\begin{figure}[t]
    \centering
    \begin{subfigure}[b]{0.24\textwidth}
        \centering
        \includegraphics[width=.5\textwidth]{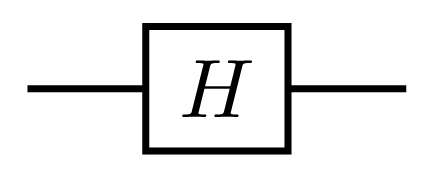}
    \end{subfigure}
    \begin{subfigure}[b]{0.24\textwidth}
        \centering
        \includegraphics[width=.5\textwidth]{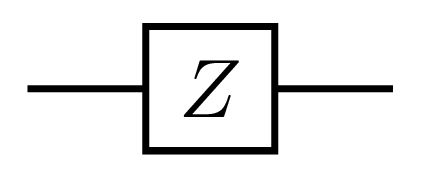}
    \end{subfigure}
    \begin{subfigure}[b]{0.24\textwidth}
        \centering
        \includegraphics[width=.5\textwidth]{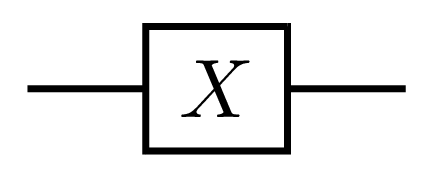}
    \end{subfigure}
    \begin{subfigure}[b]{0.24\textwidth}
        \centering
        \includegraphics[width=.5\textwidth]{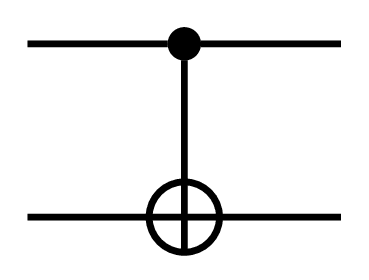}
    \end{subfigure}
    \caption{Gate representations of the $H, Z, X$ and $CNOT$ operations respectively.}
    \label{fig:qopgates}
    \Description{Graphical quantum gates for Hadamard ($H$), phase ($Z$), NOT ($X$) and controlled-NOT ($CNOT$) operations.}
\end{figure}

\begin{figure}[t]
    \centering
    \includegraphics[width=.45\textwidth]{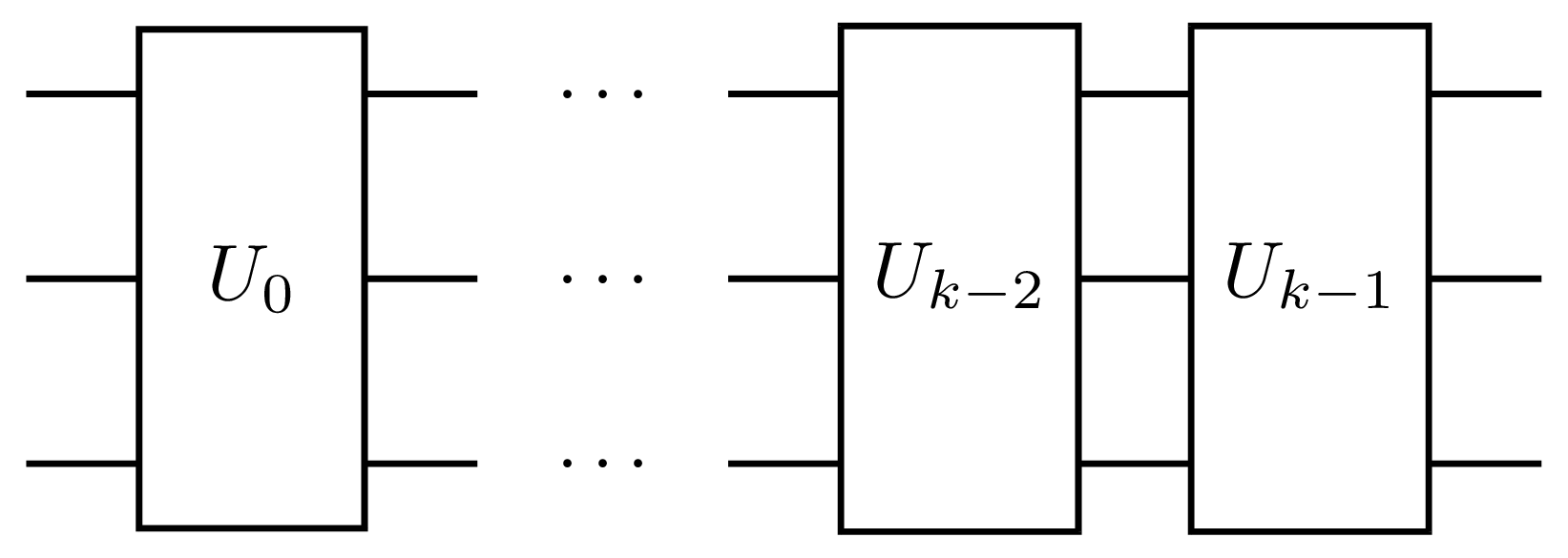}
    \caption{Example of a quantum circuit.}
    \label{fig:qcirc:ex}
    \Description{A three qubit system with operations $U_0, \dots, U_{k-2}, U_{k-1}$.}
\end{figure}

\subsection{Barrier Certificates for Real Systems}
In this section, we cover the fundamentals of barrier certificates for discrete-time real  systems.
A dynamical system works in a state space $X \subseteq \mathbb{R}^n$, for some natural number $n$, and evolves according to
\begin{equation*}
    x_{t+1} = f(x_t),\quad f:\mathbb{R}^n \to \mathbb{R}^n,
\end{equation*}
where $x_t \in X$ denotes the state of a system at time $t \in \posint{}$.

Whilst there are several problems to consider with dynamical systems, such as reachability~\cite{Mitchell07}, we are interested in the problem of safety.
\begin{definition}[Safety]
\label{def:safety}
A system, $x_{t+1} = f(x_t)$, evolving over $X \subseteq \mathbb{R}^n$ is considered safe if it cannot reach the unsafe set, $X_u \subseteq X$, from the initial set, $X_0 \subseteq X$.
That is for all $t \in \posint{}$ and $x(0) \in \init{X}$, then $x(t) \notin \unsafe{X}$.
\end{definition}

A barrier certificate is a function $B : X \to \mathbb{R}$ that satisfy the conditions
\begin{subequations}
\label{eq:gen}
\begin{align}
    B(x) \leq \kappa &, \forall x \in \init{X};
    \label{eq:gen:init}\\
    B(x) > \kappa &, \forall x \in \unsafe{X};
    \label{eq:gen:unsafe} \\
    x_0 \in \init{X} \implies B(x_t) \leq \kappa &, \forall t \in \posint{};
    \label{eq:gen:barrier}
\end{align}
\end{subequations}
for some $\kappa \in \mathbb{R}$.
These conditions do two things to the system:
(i) they separate the state space into an over-approximate reachable region and an unsafe region (through conditions~\eqref{eq:gen:init} and \eqref{eq:gen:unsafe}),
(ii) they ensure that the system cannot enter the unsafe region (through condition~\eqref{eq:gen:barrier}).
By finding a function $B$ that meets these conditions, we can therefore have safety of the system.
\begin{remark}
    We have used the term ``barrier certificate'' to refer to the functions satisfying certain inequalities that guarantee safety of the system. Previous works, e.g., \cite{ames2019control,nejati2022compositional,kordabad2024control}, have also used the term ``barrier function'' for functions that obey a similar set of conditions as to those in \eqref{eq:gen}.
\end{remark}

Conditions~\eqref{eq:gen:init} and \eqref{eq:gen:unsafe} are easily achievable by many functions, but it is usually condition~\eqref{eq:gen:barrier} that restricts the functions that we can use.
It is standard that barrier certificates imply the last condition through some other condition.
For example, consider the following definition:
\begin{definition}[Discrete-Time Barrier Certificate~\cite{Anand21}]
\label{def:discretebarrier}
For a discrete-time system $x_{t+1} = f(x_t)$, $X \subseteq \mathbb{R}^n$, $X_0 \subseteq X$ and $X_u \subseteq X$, a function $B: \mathbb{R}^n \to \mathbb{R}$ that obeys the following conditions:
\begin{subequations}
\label{eq:discretebarrier}
\begin{align}
    B(x) \leq 0 &, \forall x \in \init{X}; \\
    B(x) > 0 &, \forall x \in \unsafe{X}; \\
    B(f(x)) - B(x) \leq 0 &, \forall x \in X;
\end{align}
\end{subequations}
is a discrete-time barrier certificate.
\end{definition}

It should be trivial to see how the constraints given in \eqref{eq:discretebarrier} imply the conditions in~\eqref{eq:gen} with $\kappa = 0$.

\begin{remark}
    Barrier certificates can additionally be used as a technique for verifying reachability properties of a system.
    The techniques that we use in this article for generalizing discrete dynamical systems to complex systems can be applied to continuous or stochastic dynamical systems, other barrier certificates for safety, or even barrier certificates for reachability properties.
\end{remark}

\section{Quantum Circuits as Dynamical Systems}
\label{sec:qcirc2dynam}
We begin by introducing the dynamical system we consider to represent our quantum circuits.

\begin{definition}[Discrete-time complex-space system]
    A {discrete-time} {complex-space} system is a tuple $S = \timesys$, where
    \begin{itemize}
        \item $\sysspace \subseteq \cspace$ is the continuous (complex-valued) state space;
        \item $\init{\sysspace} \subseteq \sysspace$ is the set of initial states;
        \item $F$ is a {\em finite} set of functions that contain all the possible dynamics the system can perform;
        \item and $f : \mathbb{Z}_{\geq 0} \to F$  assigns at each time step the dynamics of the system (\ie $f$ is total).
    \end{itemize}
    At each time step, $t$, the state of the system is $z_t$ and the dynamics of the system is defined by
    \begin{equation*}
    z_{t+1} = f(t)(z_t) = f_t(z_t).    
    \end{equation*}
\end{definition}

Conversion of quantum circuits into the dynamical system given above is straightforward.
For an $n$-qubit system, quantum states are restricted to points on the unit circle of $\mathbb{C}^{2^n}$, \ie{} $\sysspace = \{z \in \mathbb{C}^{2^n} : \sum_j \abs{(z)_j}^2 = 1\}$.
We assume that for any quantum circuit, there is some initial state $\ket{\phi} = \sum_{j=0}^{2^n} (z)_j \ket{j}$ or a set of initial points that can be chosen from.
One can include error in the initial state through noise that occurs in preparing the initial state.
Thus, $\init{\sysspace}$ can be specified by the user depending on how noisy preparing the initial state is.
Finally, say the quantum circuit is of the form $U = U_{k-1} \dots U_1 U_0$, where $U_0, U_1, \ldots$ are unitary operations.
Then we simply have that $F = \{U_t\}_{t=0,\dots,k-1}$, and $f(t) = U_t$ for $t < k$ and $f(t) = I_{2^n}$ otherwise.

\begin{figure}[t]
    \centering
    \includegraphics[width=0.6\textwidth]{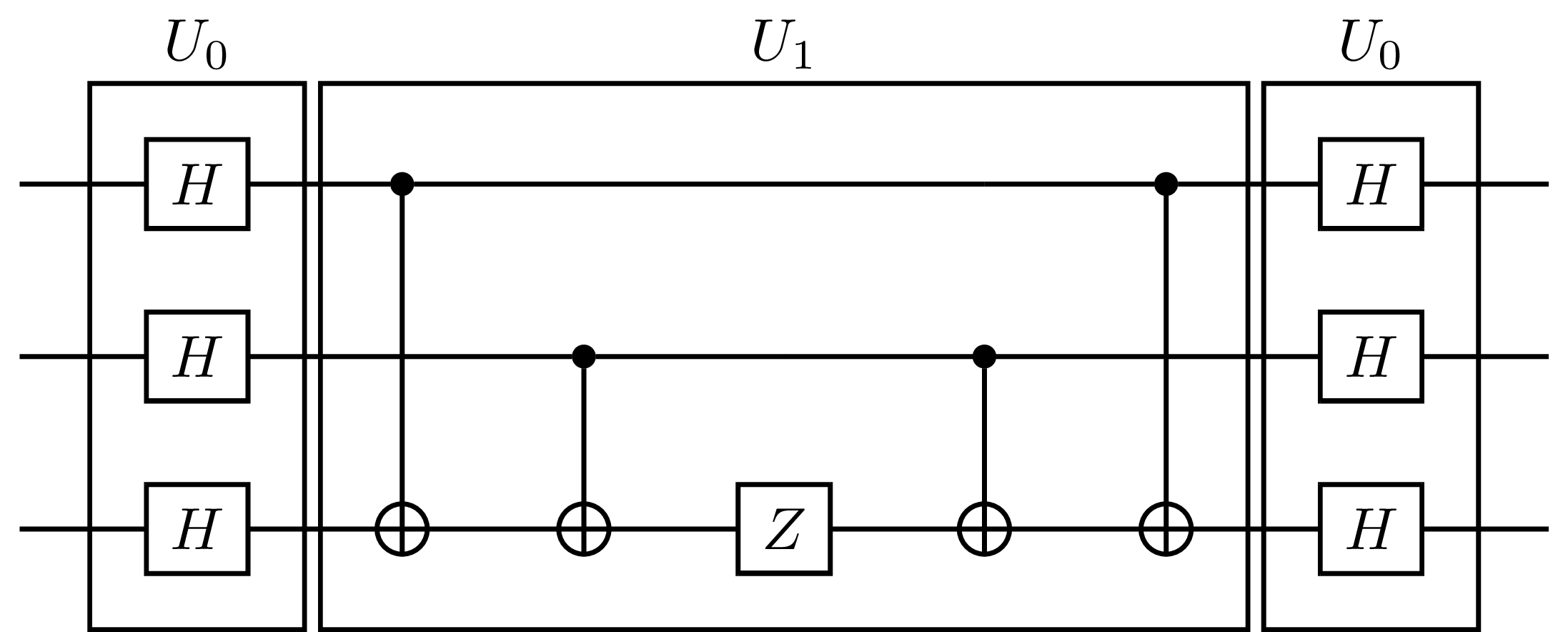}
    \caption{Quantum circuit with 3 unitary operations consisting of grouped gates.}
    \label{fig:unitary2dynamical}
    \Description{A 3-qubit system consisting of operation group consisting of $H^{\otimes 3}$, a collection of $CNOT$ operation and a $Z$ operation, and $H^{\otimes 3}$.}
\end{figure}

There is flexibility in what level the unitary operations are chosen, since one can let unitary operations be individual gates or they can be a combination of several gates at different depths.
For example, the quantum circuit in Figure~\ref{fig:unitary2dynamical} is equivalent to the system
$S = \timesys$ where
\begin{align*}
    & \sysspace = \{z \in \mathbb{C}^{8} : \sum_j \abs{(z)_j}^2 = 1\}; \\
    & \init{\sysspace} = \{z \in Z : \abs{(z)_0}^2 \geq 0.99, \Im{(z)_0} = 0\}; \\
    & F = \{U_0, U_1\}; \text{ and} \\
    & f_t(z) = \begin{cases}
    U_0 z & \text{for } t = 0, 2, \\
    U_1 z & \text{for } t = 1, \\
    z & \text{otherwise;}
    \end{cases}
\end{align*}
and the actual unitary transformation corresponding to $U_0$ and $U_1$ can be easily deduced from Figure \ref{fig:unitary2dynamical}.
In this example, $\init{\sysspace}$ corresponds to the set of initial quantum states having $\ket{0}$ be measured with at least 99\% probability.

\section{Complex Barrier Certificates}
\label{sec:complexbarriers}
We introduce the notion of safety for our complex system.
\begin{definition}[Safety]
Let $S=\timesys$ be a discrete-time complex-space system and $\unsafe{\sysspace} \subseteq \sysspace$ denote the unsafe set.
Then $S$ is safe if for all $\init{z} \in \init{\sysspace}$, we have that $z_t \notin \unsafe{\sysspace}$ for any $t \in \mathbb{Z}_{\geq 0}$.
\end{definition}

To solve the problem of safety using barrier certificates we combine the ideas behind barrier certificates for hybrid systems \cite{Prajna04} and $k$-inductive barrier certificates \cite{Anand21}, adapting them to our dynamical system as well as the complex domain.
Additionally, we adapt the certificate to handle different dynamics that occur during evolution.

The barrier certificates for~\cite{Prajna04} are capable of handling systems with dynamics that occur on both a continuous and discrete space.
For our system $S = \timesys$, these spaces corresponds to our complex space, $\sysspace$, and the time steps, $\mathbb{N}_{\geq 0}$, respectively.
The $k$-inductive barrier certificates introduced in~\cite{Anand21} work on the basis that the system is allowed to evolve slightly towards the unsafe region but after $k$ steps the system will be further away from the unsafe region than before taking the $k$ steps.
This $k$-inductive property is a less strict property than, for example, the discrete-time barrier certificate given in Definition~\ref{def:discretebarrier}.
By having a less restrictive property, the set of functions satisfying this property is larger and so we adopt this property for our barrier certificates.

\begin{theorem}[$k$-Inductive Hybrid Barrier Certificate]
    \label{thm:timesafe}
    Let $S = \timesys$ and the unsafe set be $\unsafe{\sysspace}$. 
    Suppose there exists a $k$-inductive hybrid barrier certificate: a collection of functions $\{B_t(z)\}_{t\in\mathbb{Z}_{\geq 0}}$, where $B_t(z) \in \mathbb{R}$ for all $t \geq 0$, $z \in \sysspace$; and constants $k \geq \mathbb{Z}_{\geq 1}$, $\epsilon, \gamma \in \mathbb{R}_{\geq 0}$ and $d > k (\epsilon + \gamma)$ that satisfy the following equations
    \begin{subequations}
    \label{eq:barrier}
    \begin{align}
        & B_0(z) \leq 0, \forall z \in \init{\sysspace}; \label{eq:barrier:initial} \\
        & B_t(z) \geq d, \forall z \in \unsafe{\sysspace}, \forall t \in \mathbb{Z}_{\geq 0}; \label{eq:barrier:unsafe} \\
        & B_t(f_t(z)) - B_t(z) \leq \epsilon, \forall z \in \sysspace, t \in \mathbb{Z}_{\geq 0}; \label{eq:barrier:diff} \\
        & B_{t+1}(z) - B_{t}(z) \leq \gamma, \forall z \in \sysspace, t \in \mathbb{Z}_{\geq 0}; \label{eq:barrier:mode} \\
        & B_{t+k}(f_{t+k-1} ( \hspace{0cm} \dots \hspace{0cm} f_{t+1}( f_t (z)))) - B_t(z) \leq 0, \forall z \in \sysspace, t \in \mathbb{Z}_{\geq 0} \label{eq:barrier:ind} \\
        & \qquad{} \text{ such that } t = rk \text{ for } r \in \mathbb{Z}_{\geq 0}. \nonumber{}
    \end{align}
    \end{subequations}
    Then the safety of $S$ with respect to $\unsafe{\sysspace}$ is guaranteed.
\end{theorem}
\begin{proof}
    We reason by contradiction.
    Let the system have a $k$-inductive hybrid barrier certificate, \ie we have $\{B_t(z)\}, k, \epsilon, \gamma, d$ that satisfy the conditions given in \eqref{eq:barrier}, and assume that the system is not safe.
    
    Let $(z_0, \dots, z_{T})$ be a trace, where $z_{t+1} = f_t(z_t)$, that reaches an unsafe state, $z_{T} \in \unsafe{Z}$.
    Let $T= t + m$, where $t,m \in \mathbb{Z}_{\geq 0}$, $t = rk$ for some $r\in \mathbb{Z}_{\geq 0}$ and $m < k$.
    We have that $B_0(z_0) \leq 0$ and $B_{t+m}(z_{t+m}) \geq d$.
    Using \eqref{eq:barrier:diff} and \eqref{eq:barrier:mode}, we have that
    \begin{align*}
        & B_{t + m}(z_{t+m}) = B_{t+m}(f_{t+m-1}(z_{t+m-1})) \\
        & \leq B_{t + m - 1}(f_{t+m-1}(z_{t+m-1})) + \gamma \\
        & \leq B_{t + m - 1}(z_{t+m-1}) + \epsilon + \gamma \\
        & \leq \dots \\
        & \leq B_t(z_t) + m (\epsilon + \gamma) \leq B_t(z_t) + k(\epsilon + \gamma).
    \end{align*}
    Using \eqref{eq:barrier:ind} and induction, we have that $B_t(z_t) \leq B_{t-k}(z_{t-k}) \leq \dots \leq B_0(z_0)$.
    Therefore, we have
    \begin{equation*}
    B_{t+m}(z_{t+m}) \leq B_0(z_0) + k(\epsilon + \gamma) < d.
    \end{equation*}
    This is a contradiction to \eqref{eq:barrier:unsafe} and therefore $z_{T} \notin \unsafe{Z}$.
    Therefore, the system $S$ is safe.
\end{proof}
\begin{remark}
\label{rmk:condimply}
    While the barrier certificate in Definition~\ref{def:discretebarrier} is time-invariant, the barrier certificate introduced in Theorem \ref{thm:timesafe} is time-varying.
    Also, note that with certain values of $k, \epsilon$ and $\gamma$, the number of equations to satisfy can be reduced.
    For example, if $k = 1$ and $\gamma = 0$, then Equation~\eqref{eq:barrier:diff} is implied through Equations~\eqref{eq:barrier:mode} and \eqref{eq:barrier:ind} for any value of $\epsilon$.
    Whilst time-invariant barriers are easier to find, they might not work well for time-varying systems.
    We make use of $k$-inductive hybrid barrier certificates since initial tests with time-invariant barriers could not be generated for some systems.
\end{remark}

Informally, Equations~\eqref{eq:barrier:initial} and \eqref{eq:barrier:unsafe} cover the separation of the initial and unsafe regions, providing a buffer between them.
Equations~\eqref{eq:barrier:diff} and \eqref{eq:barrier:mode} ensures that as the barrier evolves over time the system does not grow too much, both when applying dynamical operations and when swapping between dynamical operations.
Finally, Equation~\eqref{eq:barrier:ind} is the inductive condition that ensures that after $k$ steps, the system remains in a negative space.
These conditions combined together mean that as the system evolves, the barrier may enter the buffer space between the initial and unsafe region for up to $k$ steps but will trend to grow negatively and therefore never enter the unsafe region.

One of the challenges with Theorem~\ref{thm:timesafe} is that the functions, $B_t$, have complex variables but are required to return a real value.
Therefore, we need to restrict the functions to specific classes that can be easily defined and will be useful for finding a barrier.
To do this, we adapt a definition used in
\cite{Lewis23}.
\begin{definition}[Conjugate-flattening Function~\cite{Lewis23}]
    A function $P(z): \mathbb{C}^n \to \mathbb{R}$ is a conjugate-flattening function if $P(z) = p(z, \overline{z})$ for some complex function $p: \mathbb{C}^n \cross \mathbb{C}^n \to \mathbb{C}$ such that $p(z, \overline{z}) \in \mathbb{R}$ for all $z \in \mathbb{C}^n$.
\end{definition}
\begin{definition}[Conjugate-flattening Polynomial]
    A conjugate-flattening function $P(z) = p(z, \overline{z})$ is a conjugate-flattening polynomial if $p(z, u)$ is a polynomial with complex variables and coefficients.
\end{definition}

These definitions can be used in Theorem~\ref{thm:timesafe} for the collection of functions to ensure that $B_t(z) \in \mathbb{R}$ for any $t \geq 0, z \in \sysspace$.
Being able to differentiate between conjugate-flattening functions and polynomials will be useful when we discuss the generation of barrier certificates in Section~\ref{sec:hsos:bcs}.

\begin{remark}    
A barrier certificate that returns complex values, $B_t(z) \in \mathbb{C}$, cannot be used since then one cannot reason about the ordering of the barriers with respect to the equations in \eqref{eq:barrier}.
Additionally, using aspects of a complex number (\ie modulus, real or imaginary parts) can be represented by a conjugate-flattening function in some way, \eg $\abs{z}^2 = z \overline{z}$.
\end{remark}

\section{Computation of Barrier Certificates through Hermitian Sum of Squares}
\label{sec:generation}
There exist several approaches for computing a barrier certificate given the specification and dynamics of a system.
Most approaches are automatic and include the usage of neural networks and SMT (Satisfiability Modulo Theory) solvers to compute a barrier \cite{FOSSIL,SMT,abate2024safe}.
However, the standard approach is to make use of sum of squares (SOS) optimization in order to find a suitable barrier~\cite{Parrilo2003}.
Unlike other techniques, SOS optimization provides an efficient method whilst remaining theoretically sound. In this section, we adapt the SOS technique to complex variables.

\subsection{Sum of Squares for Complex Numbers}
A polynomial with real variables and coefficients, $p(x)$, is referred to as a sum of squares (SOS) if $p(x) = \sum_k p_k(x)^2$, where $p_k$ are polynomials (of any degree) for $k \geq 1$.
There are two reasons that SOS polynomials are useful for generating barrier certificates:
\begin{enumerate}
    \item SOS polynomials are real and positive.
    This makes SOS polynomials useful for ensuring generated functions obey theorems for safety (such as Theorem~\ref{thm:timesafe}).
    \item There is an equivalence between SOS polynomials and positive semidefinite matrices: $p(x)$ is a SOS \emph{iff} there exists a positive semidefinite matrix $Q$ such that $p(x) = v(x)^\intercal Q v(x)$ where $v(x)$ is a vector of monomial terms.
    Such matrices can be found using semidefinite programming~\cite{Parrilo2003}.
\end{enumerate}
Combining these two properties gives us an efficient and sound method for finding real polynomial barriers.

A function with complex variables and coefficients, $p$, may produce complex values, \ie $p: \mathbb{C}^n \to \mathbb{C}$.
This means that $p$ may not produce only positive, or even real, values.
Thus, we need a method to ensure $p(z)$ is real and positive for any $z \in \mathbb{C}^n$.
This can be done by considering a variation of sum of squares.

\begin{definition}[Hermitian Sum of Squares~\cite{HSOS}]
    A complex function, $p(z) : \mathbb{C}^n \to \mathbb{C}$ is a Hermitian sum of squares (HSOS) if $p(z) = \sum_k p_k(z)\overline{p_k(z)}$ where $p_k: \mathbb{C}^n \to \mathbb{C}$ are complex polynomials and $k \geq 1$.
    \label{def:hsos}
\end{definition}

\begin{remark}
Note that the main difference between HSOS and SOS is that the conjugate ($\overline{z}$) is used for HSOS.
If we restrict a HSOS polynomial to have only real variables and coefficients, then we get the standard SOS definition as $p_k(x) \overline{p_k(x)} = p_k(x)^2$.
\end{remark}

The properties previously described for SOS polynomials hold for HSOS polynomials with some slight modifications for the complex domain.

\begin{proposition}
    Let $p(z)$ be a HSOS.
    Then $p(z) \in \mathbb{R}$ for all $z \in \mathbb{C}^n$ and $p(z)$ is positive ($p(z) \geq 0$).
    \label{prop:hsos:positive}
\end{proposition}
\begin{proof}
    \begin{equation*}
    p(z) = \sum_k p_k(z)\overline{p_k(z)} = \sum_k \abs{p_k(z)}^2
    \end{equation*}
    It is clear that the right-hand side is a real value and is also positive.
    \qed
\end{proof}

Additionally, there is a notion of positive semidefinite matrices for complex vectors as well.

\begin{definition}[Positive semidefinite~\cite{Horn_Johnson_2012}]
A complex $n \times n$ matrix $Q$ is (Hermitian) positive semidefinite if $z^\dagger  Q z \geq 0, \forall z \in \mathbb{C}^n$.
\end{definition}

With these definitions, we can now show an equivalence between HSOS polynomials and (Hermitian) positive semidefinite matrices.

\begin{proposition}
    Let $p(z): \mathbb{C}^n \to \mathbb{C}$ be a conjugate-flattening polynomial of degree $2d$.
    Then, $p(z)$ is a HSOS iff there exists a (Hermitian) positive semidefinite matrix, $Q$, such that $p(z) = v(z)^\dagger  Q v(z)$, where $v(z)$ is a vector of all monomials of degree less than $d$.
    \label{prop:hsos:semidefinite}
\end{proposition}
\begin{proof}
    Firstly, assume that $p(z)$ is a HSOS.
    Therefore, $p(z) = \sum_k p_k(z) \overline{p_k(z)}$ where $p_k$ is a polynomial.
    Since $p$ has degree $2d$ and $\deg(p_k(z)) = \deg(\overline{p_k(z)})$, then each $p_k$ is of degree up to $d$.
    Therefore, $p_k(z) = B_k v(z)$, where $B_k$ is a (row) vector of complex coefficients.
    
    Let $B$ be the matrix whose rows are $B_k$ and elements are $(B)_{kl}$.
    
    By expanding $p(z)$, we have that
    
    \begin{equation*}
        \begin{aligned}
            p(z) &= \sum_k B_k v(z) \overline{B_k v(z)}
            \\
            & = 
            \sum_k \sum_{l,m} (B_k v(z))_l (\overline{B_k v(z)})_m
            \\
            & = 
            \sum_k \sum_{l,m} (B)_{kl} v(z)_l (\overline{B})_{km} (\overline{v(z)})_m
            \\
            & = 
            \sum_k \sum_{l,m} (B^\dagger )_{mk} (B)_{kl} (v(z))_l (\overline{v(z)})_m
            \\
            & = 
            \sum_{l,m} \bigg( \sum_k (B^\dagger )_{mk} (B)_{kl} \bigg) (v(z))_l \overline{(v(z))_m}
            \\
            & = 
            \sum_{l,m} (B^\dagger  B)_{ml} (v(z))_l \overline{(v(z))_m}
            \\
            & =
            v(z)^\dagger  (B^\dagger B) v(z).
        \end{aligned}
    \end{equation*}
    Since $B$ is a complex matrix, then $Q = B^\dagger B$ is positive semi-definite.
    
    Now if we assume that $p(z) = v(z)^\dagger  Q v(z)$ with $Q$ being positive semi-definite, then there exists a complex matrix $B$ such that $Q=B^\dagger B$.
    Simply by performing the summation in reverse gives us that $p(z) = \sum_k p_k(z) \overline{p_k(z)}$ where $p_k(z) = B_k v(z)$.
\end{proof}

Since properties of SOS polynomials are shared by HSOS polynomials, we can make use of computation techniques used to find barrier certificates for real systems and adapt them for complex systems.

\subsection{Semi-algebraic Sets}
In the usual method for computing barrier certificates, the sets used (\eg the unsafe and initial sets) are assumed to be semi-algebraic \cite{Bochnak98}, \ie{} they can be described using vectors of polynomials.
For example, if $X \subseteq \mathbb{R}^n$ is semi-algebraic, we can write
\begin{equation*}
    X = \{x \in \mathbb{R}^n: g(x) \geq 0 \},
\end{equation*}
where $g(x)$ is a vector of polynomials and the ordering is applied element-wise ($g_j(x) \geq 0$ for all $j$).

This definition does not immediately hold for complex variables, since complex values cannot be ordered.
Thus, semi-algebraic sets for complex numbers must be defined by a vector of conjugate-flattening polynomials, \ie{} $g(z) \in \mathbb{R}^n$ for all $z \in \mathbb{C}^n$.
We introduce a definition for complex semi-algebraic sets.
\begin{definition}[Complex semi-algebraic set]
    A set $Z \subseteq \mathbb{C}^n$ is complex semi-algebraic if
    \begin{equation*}
        Z = \{z \in \mathbb{C}^n: g(z) \geq 0 \};
    \end{equation*}
    for some $g$ being a vector of conjugate-flattening polynomials and the ordering ($g(z) \geq 0$) is applied element-wise.
\end{definition}
In our case, we assume that the sets $\sysspace, \init{\sysspace}$, and $\unsafe{\sysspace}$ are (complex) semi-algebraic and use the vectors $g_I$, $\init{g}$, and $\unsafe{g}$ to describe their elements respectively.

\subsection{HSOS Equations for Barrier Certificate Constraints}
\label{sec:hsos:bcs}
We can now state the HSOS equations that need to be computed to get a function that satisfies the constraints of a barrier certificate.
\begin{lemma}
    \label{lem:hsossafe}
    Let $S = \timesys$; $g_I$, $g_0$, $g_u$ be given vectors of conjugate-flattening polynomials describing $\sysspace, \init{\sysspace},$ and $\unsafe{\sysspace}$, respectively.
    Suppose there exists a collection of (conjugate-flattening) polynomials $\{B_t(z) = b_t(z, \overline{z})\}_{t\in \mathbb{Z}_{\geq 0}}$; positive numbers $k \in \mathbb{Z}_{\geq 0}$, $\epsilon, \gamma \in \mathbb{R}_{\geq 0}$, $d > k (\epsilon + \gamma)$; and vectors of Hermitian sum of squares $\lambda_{\text{U};t}(z)$, $\lambda_{\text{Init}}(z)$, $\lambda_t(z)$, $\lambda_{t,t'}(z)$ and $\hat{\lambda}_t(z)$ such that the expressions:
    \begin{subequations}
    \label{eq:barrierhsos}
    \begin{align}
        & - B_0(z) - {\lambda_{\text{Init}}}(z)^\intercal \init{g}(z); \\
        & B_t(z) - {\lambda_{\text{U};t}}(z)^\intercal \unsafe{g}(z) - d,  \forall t \in \mathbb{Z}_{\geq 0}; \label{eq:barrierhsos:unsafe} \\
        & - B_t(f_t(z)) + B_t(z) - {\lambda_t}(z)^\intercal g_I(z) + \epsilon, \forall t \in \mathbb{Z}_{\geq 0}; \\
        & - B_{t+1}(z) + B_t(z) - {\lambda_{t,t+1}}(z)^\intercal g_I(z) + \gamma, \forall t \in \mathbb{Z}_{\geq 0}; \\
        & - B_{t+k}(f_{t+k-1} ( \hspace{0cm} \dots \hspace{0cm} f_{t+1} ( f_{t} (z)))) + B_{t}(z) - \hat{\lambda}_{t}(z)^\intercal g_I(z),
        \label{eq:barrierhsos:kind}
        \\ & \nonumber \qquad{} \forall t \in \mathbb{Z}_{\geq 0} \text{ such that } t = rk \text{ for } r \in \mathbb{Z}_{\geq 0}; 
    \end{align}
    \end{subequations}
    are Hermitian sum of squares.
    Then $B$ satisfies Theorem~\ref{thm:timesafe} and the safety of $S$ is guaranteed.
\end{lemma}
\begin{proof}
    We show that expression~\eqref{eq:barrierhsos:unsafe} being HSOS implies that $B$ satisfies constraint~\eqref{eq:barrier:unsafe} of Theorem \ref{thm:timesafe}.
    Firstly, note that since $\lambda_{\text{U};t}$ is a vector of HSOSs, then $\lambda_{\text{U};t}(z) \geq 0, \forall z \in \sysspace$.
    Additionally, for $\unsafe{z} \in \unsafe{\sysspace}$, we have that $g(\unsafe{z}) \geq 0$.
    Therefore, $\lambda_{\text{U};t}(\unsafe{z})^\intercal g(\unsafe{z}) \geq 0$ and we must have that $B_t(\unsafe{z}) - d \geq 0$ as Equation~\eqref{eq:barrierhsos:unsafe} is HSOS.
    Thus, constraint~\eqref{eq:barrier:unsafe} is satisfied.
    
    Similar reasoning can be applied to the other expressions in \eqref{eq:barrierhsos} and their respective counterparts in \eqref{eq:barrier}.
    This results in the conditions in Theorem~\ref{thm:timesafe} being satisfied and, therefore, safety is guaranteed.
\end{proof}

In the same way HSOS is related to SOS, Lemma~\ref{lem:hsossafe} is similar to theorems and lemmas for real dynamical systems, with the major difference being the usage of HSOS equations rather than SOS equations to allow us to use complex variables.

\begin{remark}
    Note that the barrier certificate generated is a collection of conjugate-flattening polynomials, but the definition of $k$-inductive hybrid barrier certificates given in Theorem~\ref{thm:timesafe} does not require the barrier certificate to use polynomials.
    Simply by restricting the barrier certificate to conjugate-flattening polynomials allows us to easily adapt the SOS methods to the complex domain while retaining a sufficient condition for safety of the system.
\end{remark}

\subsection{Algorithm for finding a Barrier Certificate}
\begin{algorithm}[t]
\caption{Finding a barrier certificate}\label{alg:barrier}
\KwIn{Constants: $k, \delta, \epsilon, \gamma$; discrete-time complex-space dynamical system: $S = \timesys$; vectors describing associated semi-algebraic sets: $g_I, g_0, g_u$}
Set $d > k(\epsilon + \gamma)$.

\textbf{Setup symbolic function:} Set $B(z) = b(z,\overline{z})$ to be a parameterized conjugate-flattening polynomial of degree up to $\delta$.

\textbf{Define symbolic lambda polynomials:} parameterized conjugate-flattening polynomials $\lambda_{\text{Init}}(z)$, $\lambda_{U;t}(z)$, $\lambda_t(z)$, $\lambda_{t,t'}(z)$ and $\hat{\lambda}_t(z)$ for each $t, t'$.

\textbf{Add to HSOS solver:} Add the various $\lambda$ polynomials, and the equations in \eqref{eq:barrierhsos} (using the $B, \lambda$ polynomials and constants $k, \epsilon, \gamma, d$) as HSOS constraints to the HSOS solver.

\textbf{Run the HSOS solver:} convert the HSOS equations of \eqref{eq:barrierhsos} into an equivalent semidefinite program (via Proposition~\ref{prop:hsos:semidefinite}) and solve using a semidefinite programming algorithm.

\eIf{feasible}{
\Return $B(z)$ with coefficients set to the values from the HSOS solver
}
{
\textbf{error:} conjugate-flattening HSOS polynomial barrier does not exist for system $S$ with inputs given
}
\end{algorithm}
We now provide an algorithm for finding a barrier certificate for a discrete-time complex-space system $S = \timesys$, given in Algorithm~\ref{alg:barrier}.
In order to compute a barrier, we consider $B(z)$ as a $\delta$-degree conjugate-flattening polynomial.
We parameterize the coefficients of $B$ as a vector of complex values $\beta \in \mathbb{C}^\kappa$, where $\kappa \approx \sum_{j=0}^\delta (2n)^{j} = \frac{1 - (2n)^{\delta+1}}{1 - 2n}$ is the number of coefficients in a $\delta$-degree polynomial with $2n$ variables ($2n$ comes from the fact that we need to consider $z$ and $\overline{z}$ terms).
We input $B(z)$ with its parameterized coefficients $\beta$ into the equations given in \eqref{eq:barrierhsos}.
These equations are then given to an appropriate HSOS solver, which attempts to find appropriate values for $\beta$.

For specific hyper-parameters a barrier may not be produced, hence why the HSOS may return \emph{infeasible} and there is an error branch (line 9) in Algorithm~\ref{alg:barrier}.
This can be due to several reasons, for example the barrier may have a degree that is higher than $\delta$ or the values of $\epsilon$ and $\gamma$ need to be changed.
The system may simply be unsafe.
Alternatively, the required barrier for the system and unsafe region given may need to be non-polynomial, in which case the HSOS technique cannot be used to generate a barrier.

We note that it is possible to adapt Algorithm~\ref{alg:barrier} such that some of the constants could instead be added to the semidefinite program as constraints.
The constants $\epsilon$ and $\gamma$ could be added as to the objective function of the semidefinite program, and $d$ can additionally be added as a constraint (with $d > k(\epsilon + \gamma)$).
This would allow Algorithm~\ref{alg:barrier} to find a barrier that is efficient in its separation of regions.


\section{Case Studies}
\label{sec:casestudies}
We now show how the theory that we have developed can be used in practice.
We provide details of our implementation of a toy HSOS solver and demonstrate its usage for barrier certificates on some example quantum circuits.

\subsection{Experimental Setup}
For our experiments we provide the dynamical system $f_t(z)$; functions describing the semi-algebraic sets: $g_I, \init{g}, \unsafe{g}$; parameters $k, \epsilon, \gamma$ from Lemma~\ref{lem:hsossafe}; and the maximum degree for the barrier certificate, represented by $\deg(B)$.

\subsubsection{HSOS Solver}
Standard SOS solvers, such as SOSTOOLS~\cite{sostools}, use semidefinite programs to compute the polynomials in the summation while the solver handles the conversion to and from SOS equations.
Unfortunately, as far as we are aware, there does not exist any HSOS solver or even a complex semidefinite optimiser that makes use of the speedups shown in \cite{Gilbert17} as current semidefinite solvers that allow for complex programs convert them to equivalent real programs.

With the equivalence shown between HSOS and positive semidefinite matrices in Proposition~\ref{prop:hsos:semidefinite}, we have developed a simple HSOS solver in Python that converts the HSOS equations into (complex) semidefinite programs.
We make use of the SymPy~\cite{sympy} package to represent polynomials, and the semidefinite programs are solved using PICOS~\cite{PICOS}, a Python interface to relevant semidefinite solvers, giving appropriate coefficients for the HSOS equation.\footnote{The complex semidefinite programs are converted into real programs by PICOS.}

\subsubsection{Verification}
To improve the soundness of the implementation of Algorithm~\ref{alg:barrier} and numerically ensure the correctness of the barrier certificate generated, we make use of SMT solvers to formally check that the generated barrier certificate satisfies Equations \eqref{eq:barrier}.
For a collection of barriers, $\{B_t(z)\}_t$ that are generated, we encode the equations in \eqref{eq:barrier} as proof obligations and look for a counter-example.
We additionally encode that the barrier must produce real values.

The equations are encoded by removing the existential quantifier ($\forall$), replacing set membership ($\in$) with satisfaction of the relevant semi-algebraic polynomial vector ($g$), and negating the rest of the statement.
For example, Equation~\eqref{eq:barrier:initial} ($\init{B}(z) \leq 0, \forall z \in \init{\sysspace}$) is encoded by the proof obligation $\init{B}(z) > 0 \land \init{g}(z) \geq 0$ (again the ordering of $\init{g}(z)$ is applied element-wise).

We make use of the Z3~\cite{Z3} Python package to setup the proof obligations, and dReal~\cite{dReal} to verify them.\footnote{We include functionality to verify using Z3 as well.}
If a proof obligation is \emph{unsatisfiable}, then that means the relevant equation is true.
However, if we receive \emph{satisfiable} or \emph{$\delta$-satisfiable} (in the case of dReal), then we have a counter-example and so the relevant equation is not satisfied (in the case of dReal, being a $\delta$-sat solver, the equation may not be actually satisfied).
We check each equation and barrier certificate function in a separate call to the SMT solver as this makes it faster to check each equation and we can find out why a certain barrier fails if we get a counter-example.

\subsubsection{Device Details}
The experiments were performed on a laptop with an Intel(R) Core(TM) i5-10310U CPU @ 1.70GHz x 8 cores processor and 16GB of RAM using Ubuntu 20.04.3 LTS.

\subsubsection{Code Availability}
The code for Algorithm~\ref{alg:barrier}, our conversion of HSOS to semidefinite programs and verifying generated barrier certificates is available at \url{https://github.com/marco-lewis/discrete-quantum-bc}.

\subsection{Phase (Z) Gate}
\label{sec:case:z}
We begin with a simple example by considering the $Z$-gate introduced in Section~\ref{sec:background:quantum}.
Consider a simple circuit with one qubit, $\sysspace = \mathbb{C}^2$, that applies the $Z$-gate repeatedly, with dynamics described by
\begin{equation*}
    f_t(z) = Zz.
\end{equation*}
We specify
\begin{equation*}
\begin{aligned}
    & \init{\sysspace} = \{ z \in \sysspace : \abs{(z)_0}^2 \geq 0.9 \}, \text{ and}\\
    & \unsafe{\sysspace} = \{ z \in \sysspace : \abs{(z)_1}^2 \geq 0.2\},
\end{aligned}
\end{equation*}
as the initial and unsafe region respectively.
These regions can be thought of quantum states with a certain state being measured with a set probability.
Here, $\init{\sysspace}$ is the region where $\ket{0}$ is likely to be measured with at least 90\% probability (similarly $\unsafe{\sysspace}$, $\ket{1}$ and 20\% respectively).

By using Algorithm~\ref{alg:barrier} with $k=1, \epsilon=0.01, \gamma=0.01$; we find the barrier (rounded to 3 d.p.)
\begin{equation*}
    B(z) = 4.453 -0.848 {(z)_{0}}^2 - 3.871 (z)_{0} \overline{(z)_{0}} + 2.274 (z)_{1} \overline{(z)_{1}} - 0.848 \overline{(z)_{0}}^2
\end{equation*}
separates the two regions and ensures safety.

\subsection{NOT (X) Gate}
\label{sec:case:x}
In this example, we consider a $X$-gate that is being applied repeatedly to a set of $n$ qubits, $\sysspace = \mathbb{C}^n$.
The dynamical system is described by
\begin{equation*}
    f_t(z) = X^{\otimes n}z,
\end{equation*}
and we have the initial and unsafe region being
\begin{equation}
\label{eq:x:init}
\begin{aligned}
    & \init{\sysspace} = \{ z \in \sysspace : \frac{1}{2^n} - err \leq \abs{(z)_j}^2 \leq \frac{1}{2^n} + err \}, \text{ and} \\
    & \unsafe{\sysspace}^p = \{ z \in \sysspace : \abs{(z)_p}^2 \geq \frac{1}{2^n} + 2 err\},
\end{aligned}
\end{equation}
where $err = \frac{1}{10^{n+1}}$ and $p \in \{0, \dots, 2^n - 1\}$.
The initial region represents the set of quantum states that is close to the uniform superposition of quantum states, $\ket{+}^n = \frac{1}{\sqrt{2^n}}\sum_{j=0}^{2^n} \ket{j}$.
The unsafe region determines that the system should avoid exiting the region for some target state, $\ket{p}$.

The barrier generated for this system with $n=2, p=0, k=2, \epsilon=0.01, \gamma=0$; is
\begin{equation*}
\begin{aligned}
B(z) = & - 51.56 + 204.89 (z)_0 (\overline{z})_0 + 0.03 (z)_1 (\overline{z})_1 + 0.03 (z)_2 (\overline{z})_2 + 0.03 (z)_3 (\overline{z})_3 \\
& - 0.03 (z)_1 (\overline{z})_2 - 0.03 (z)_2 (\overline{z})_1 - 0.03 (z)_0 (\overline{z})_3 - 0.03 (z)_3 (\overline{z})_0,
\end{aligned}
\end{equation*}
with which safety is ensured.

\subsection{Alternating X and Z Gates}
\label{sec:case:h}
\begin{figure}[t]
    \centering
    \includegraphics[width=.3\textwidth]{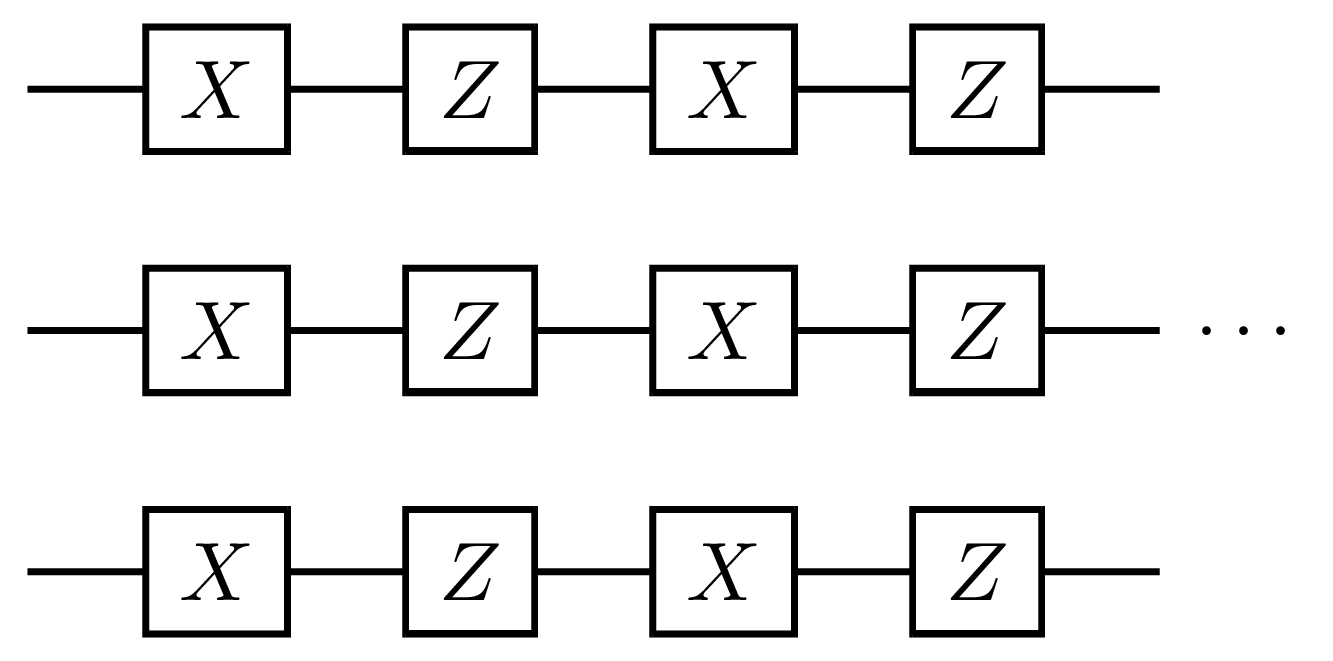}
    \caption{Quantum circuit for alternating between $X$ and $Z$ gates with 3 qubits.}
    \label{fig:xz-circuit}
    \Description{A 3-qubit circuit consisting of $X^{\otimes 3}, Z^{\otimes 3}, X^{\otimes 3}, Z^{\otimes 3}, \dots$}
\end{figure}

Here we consider a quantum circuit that alternates between the two operations used in our previous example.
The dynamical system is described by
\begin{equation*}
    f_t(z) = \begin{cases}
        X^{\otimes n}z, & \text{if $t$ is even;} \\
        Z^{\otimes n}z, & \text{otherwise.}
    \end{cases}
\end{equation*}
The associated quantum circuit, with $n=3$, is given in Figure~\ref{fig:xz-circuit}.

We specify our behavior to be such that the initial region is near one of the basis states and we want the system to behave in a way such that the basis state either occurs with high probability or low probability.
This behavior is specified by the regions
\begin{align*}
    & \init{\sysspace}^{p} =  \{ z \in \sysspace : \abs{(z)_p}^2 \geq 0.9 \}, \text{ and} \\
    & \unsafe{\sysspace}^{p} = \{ z \in \sysspace : 0.2 \leq \abs{(z)_p}^2 \leq 0.8 \},
\end{align*}
for $p \in \{0, 1, \dots, 2^n - 1\}$.

The generated functions for the barrier certificate that ensure the safety of the system with $n=2, k=2, \epsilon=0.01, \gamma=0.01, p=0$, are
\begin{align*}
    B_0(z) = &
    0.9117 - 1.0307 (z)_0 (\overline{z})_0 - 0.0095 (z)_0 (\overline{z})_3 + 0.0219 (z)_1 (\overline{z})_1 \\
    & + 0.0011 (z)_1 (\overline{z})_2 + 0.0011 (z)_2 (\overline{z})_1 - 0.0004 (z)_2 (\overline{z})_2 \\
    & - 0.0095 (z)_3 (\overline{z})_0 + 0.0066 (z)_3 (\overline{z})_3,
    \\
    B_1(z) = &
    0.902 -1.0212 (z)_0 (\overline{z})_0 - 0.0136 (z)_0 (\overline{z})_3 + 0.0289 (z)_1 (\overline{z})_1 \\
    & + 0.0058 (z)_1 (\overline{z})_2 + 0.0058 (z)_2 (\overline{z})_1 + 0.0083 (z)_2 (\overline{z})_2 \\
    & - 0.0136 (z)_3 (\overline{z})_0 + 0.0136 (z)_3 (\overline{z})_3,
\end{align*}
where $B_0$ is the function for the dynamics evolving according to $X$ and $B_1$ for $Z$ respectively.

\subsection{Grover's Algorithm}
\label{sec:case:grover}
\begin{figure}[t]
    \centering
    \includegraphics[width=.7\textwidth]{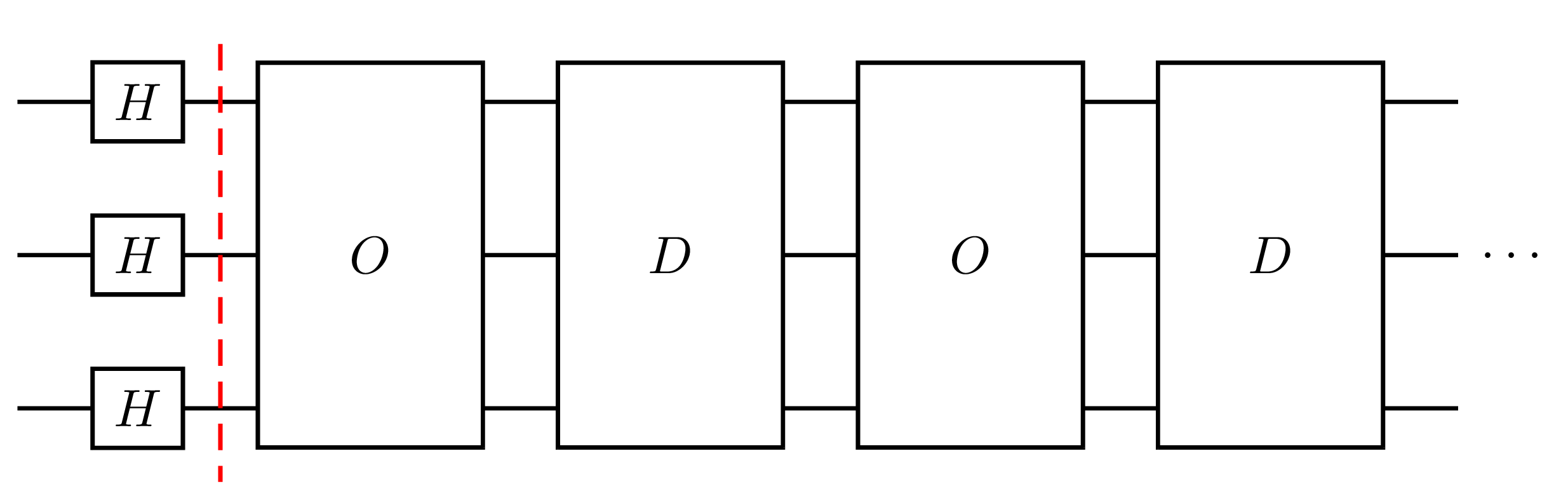}
    \caption{The quantum circuit for a 3-qubit version of Grover's algorithm.
    The initial Hadamard gates applied to the state $\ket{000}$ prepare the quantum state for the repeated operations applied to it.}
    \label{fig:grover-circuit}
    \Description{A 3-qubit circuit consisting of the operations $H^{\otimes 3}, O, D, O, D \dots$, where $O$ is an oracle operation and $D$ is the diffusion operation.}
\end{figure}
Now we want to verify properties of Grover's search algorithm~\cite{Grover96}.
The database search problem is: given a function $h:\{0,1\}^n \to \{0,1\}$ such that $h(m) = 1$ for a unique $m \in \{0,1\}^n$ and $h(x) = 0$ for $x \neq m$, find $m$ with as few calls to $h$ as possible.\footnote{There is a version of the problem that has several marked elements, but we only consider a single marked element.}
Grover's algorithm solves this by putting the quantum state into superposition (using Hadamard gates) and then alternating between an oracle, $O \ket{x} = (-1)^{h(x)} \ket{x}$, and a diffusion step, $D = H^{\otimes n} (2 \ketbra{0^n} - I_n)H^{\otimes n}$.
\emph{I.e.,} the initial quantum state is $\ket{\phi} = \frac{1}{\sqrt{2^n}} \sum_x \ket{x}$ and the operation $G = D\cdot O$ is applied to the quantum state $r$ times, where $r$ depends on the number of qubits.
This moves the quantum state into a state where $\ket{m}$ is measured with high probability.
The circuit for Grover's algorithm is given in Figure~\ref{fig:grover-circuit}.

The evolution of the quantum state can be viewed geometrically~\cite{Dorit99, nielsen_chuang_2010} as
\begin{equation*}
    G^r \ket{\phi} = \cos{\Big(\frac{2r +1}{2}\theta \Big)} \ket{\lnot m} + \sin{\Big(\frac{2r +1}{2}\theta \Big)} \ket{m},
\end{equation*}
where $0 \leq \theta \leq \pi/2$ such that $\sin{\big(\theta\big)} = \frac{2\sqrt{2^n-1}}{2^n}$ and $\ket{\lnot m} = \frac{1}{\sqrt{2^n-1}} \sum_{x\neq m} \ket{x}$.
In a setting with no noise and $n > 1$, we see that only $\ket{m}$ can be measured with high probability no matter the value of $r$.
Any unmarked element, $x$, will have at most probability $\frac{1}{2^n - 1}$.
However, if we disturb the initial state slightly, is this still the case?
Previous work has shown that Grover's algorithm with more iterations can still return the marked state even when the initial state is not in equal superposition~\cite{Biham99} but it does not show if an unmarked state is restricted.

The dynamical system for Grover's algorithm is given as
\begin{equation*}
    f_t(z) = \begin{cases}
        Oz, & \text{if $t$ is even;} \\
        Dz, & \text{otherwise.}
    \end{cases}
\end{equation*}

We set the initial region, $\init{\sysspace}$, to be a region close to the superposition of states with a slight disturbance to its amplitude.
For an $n$-qubit system, we have
\begin{align*}
    \init{\sysspace} = & \{ z \in \sysspace : \frac{1}{2^n} - err \leq \abs{(z)_j}^2 \leq \frac{1}{2^n} + err, \\
    & -\sqrt{err} \leq \Im{(z)_j} \leq \sqrt{err} \text{ for } 0 \leq j \leq 2^{n} - 1\},
\end{align*}
where, again, $err = \frac{1}{10^{n+1}}$.
This follows the definition of the initial region given in Equation~\eqref{eq:x:init} with an additional constraint on the imaginary value.

We wish to verify that no single unmarked state is ever likely, \ie an unmarked state will never be the most likely result.
This specification, for some unmarked element $p$, is given by the region
\begin{equation*}
    \unsafe{\sysspace}^p = \{ z \in \sysspace : \abs{(z)_p}^2 \geq 0.9 \}.
\end{equation*}
The value $0.9$ is to represent the unmarked state being chosen with high probability and for simplicity's sake.

We tried specifying the quantum state for $n=2$ with several different hyper-parameters but found that no barrier could be produced with a degree less than 4.
Running the algorithm with higher degrees simply takes too long given the complexity of the initial state (see the discussion in the next section).

\begin{remark}
We additionally tried changing the dynamical system into a single equation.
We experimented defining the system starting from the initial state evolving according to the dynamics $f_t(z) = D\cdot O z$ (even steps), and also tried a system where the initial state starts after applying the oracle to the initial state and the dynamics evolve according to $f_t(z) = O\cdot D z$ (odd steps).
However, no barrier could be found for either systems using the unsafe system as described before.
\end{remark}

\subsection{Discussion}
\setlength\rotFPtop{0pt plus 1fil}
\begin{sidewaystable}
    \centering
    \caption{Average runtimes for 5 runs of each experiment.
    \emph{Experiment} refers to the relevant experiment discussed in this section;
    \emph{\# Qubits} refers to the number of qubits used;
    \emph{Target State} refers to the state that is used for safety properties (where $\ket{p}$ modifies $\init{\sysspace}^p$ and $\unsafe{\sysspace}^p$);
    \emph{Generation Time} refers to the amount of time spent generating the barrier certificate, split between time performing setup and post-processing (\emph{S\&P}) and time in PICOS (\emph{PICOS});
    and \emph{Verification Time} refers to the amount of time taken to verify the barrier using SMT solvers.
    For verification a timeout was set to 300s for each call to the SMT solver.
    Entries with a * in them faced some timeout during the verification.
    Grover experiments used an oracle function where 0 is the marked state (\ie $h(0) = 1$).}
    \bgroup
    \begin{tabular}{|c|c|c|r|r|c|r|}
        \hline
         \multirow{2}{*}{\emph{Experiment}} & \multirow{2}{*}{\emph{\# Qubits}} & \multirow{2}{*}{\emph{Target State} ($\ket{p}$)} & \multicolumn{3}{c|}{\emph{Generation Time (s)}} & \multirow{2}{*}{\emph{Verification Time (s)}} \\ \cline{4-6}
         & & & \multicolumn{1}{c|}{\emph{S\&P}} & \multicolumn{1}{c|}{\emph{PICOS}} & \emph{Status} & \\ \hline
         \multirow{6}{*}{$Z$ Gate} & \multirow{2}{*}{1} 
           & $\ket{0}$ & 2.62 & 1.72 & solved & 0.15 \\
         & & $\ket{1}$ & 2.53 & 1.56 & solved & 0.16 \\ \cline{2-7}
         & \multirow{4}{*}{2} 
           & $\ket{0}$ & 86.23 & 33.85 & solved & 0.93 \\
         & & $\ket{1}$ & 87.90 & 34.34 & solved & 21.42 \\
         & & $\ket{2}$ & 89.62 & 34.32 & solved & 22.06 \\
         & & $\ket{3}$ & 89.02 & 34.28 & solved & 4.77  \\
         \hline
         \multirow{6}{*}{$X$ Gate} & \multirow{2}{*}{1} 
           & $\ket{0}$ & 2.99 & 2.14 & solved & 0.41\\
         & & $\ket{1}$ & 3.20 & 2.48 & solved & 131.26 \\ \cline{2-7}
         & \multirow{4}{*}{2}
           & $\ket{0}$ & 248.71 & 95.96 & solved & 900.37$^{*}$ \\
         & & $\ket{1}$ & 246.20 & 96.19 & solved & 900.42$^{*}$ \\
         & & $\ket{2}$ & 247.17 & 70.88 & \textbf{unsolved} & - \\
         & & $\ket{3}$ & 241.01 & 69.34 & \textbf{unsolved} & - \\
         \hline
         \multirow{6}{*}{$X$ and $Z$ Gates} & \multirow{2}{*}{1} 
           & $\ket{0}$ & 4.31 & 7.06 & solved & 0.317 \\
         & & $\ket{1}$ & 4.53 & 7.42 & solved & 1.20 \\ \cline{2-7}
         & \multirow{4}{*}{2}
           & $\ket{0}$ & 98.45 & 297.78 & solved & 160.73 \\
         & & $\ket{1}$ & 3{,}955.97 & 579.68 & solved & 2{,}104.04$^{*}$ \\
         & & $\ket{2}$ & 3{,}890.78 & 757.57 & \textbf{unsolved} & - \\
         & & $\ket{3}$ & 3{,}996.65 & 877.00 & solved & 2{,}105.83$^{*}$ \\ \cline{2-7}
         & 3 & $\ket{0}$ & - & - & \textbf{killed} & - \\ \hline
         Grover (k=1) & 2 & $\ket{1}$ & 1{,}297.32 & 523.15 & \textbf{unsolved} & - \\ \hline
         Grover (k=2) & 2 & $\ket{1}$ & 1{,}288.85 & 590.36 & \textbf{unsolved} & - \\ \hline
         Grover (even)& 2 & $\ket{1}$ & 1{,}257.80 & 98.26 & \textbf{unsolved} & - \\ \hline
         Grover (odd) & 2 & $\ket{1}$ & 1{,}261.36 & 97.90 & \textbf{unsolved} & - \\ \hline
     \end{tabular}
     \egroup
    \label{tab:results}
\end{sidewaystable}

We extended the experiments to a higher number of qubits to see how the implementation would perform.
Details of how the $Z$ gate experiment is extended is given in Appendix~\ref{app:case:z}.
The runtimes for the experiments are shown in Table~\ref{tab:results}.

While most experiments used a 2-degree barrier, the \emph{X and Zs Gates} experiment used a 4-degree barrier for the targets $\ket{1}$, $\ket{2}$, and $\ket{3}$.
Raising the degree of the barrier has a large affect on the setup for the barrier.
Further, most of the generation time was spent setting up the various polynomials rather than running the semidefinite solver in PICOS.
This time can be reduced by reducing the number of terms in the polynomial for the barrier, meaning a change in the form of the barrier could lead to faster runtimes but at the cost of not being able to generate barriers.
An alternative approach would be to cache or store the symbolic representations of functions so that they can be easily reused.

Additionally, it was found that while the implementation could find barriers for some examples, it could not for others (noted by those that were \textbf{unsolved} or \textbf{killed}).
The barriers that were timed out during verification could also be incorrect.
This is likely due to only considering a low degree polynomial for the barrier.
However, using a higher degree polynomial would take much longer in the setup phase.
It could be the case that a non-polynomial function or a different barrier scheme needs to be used to ensure safety of some quantum systems, however this should not be the case for these simple examples.

An alternative approach would be to consider a change in variables to reduce the dimensionality of the variables.
As seen in the results for Grover's algorithm, changing the dynamics of the system to use only a single system vastly reduced the time spent solving the semidefinite program, reducing the dimensionality could reduce it further.
For instance, the dynamics of Grover's algorithm could be encoded using the geometric representation instead of the standard quantum state representation.
This would allow for a reduction in the dimension of the system and could allow some properties of quantum circuits to be verified, although the properties would need to be specified using this new basis.
As it currently stands though, while some systems can be proven safe, barrier certificates are inefficient if naively applied to quantum circuits with a high number of qubits.

Finally, while the experiments were run on a device with modest resources, PICOS failed to run a semidefinite solver an example using 3 qubits (given in the \emph{X and Z Gates} experiment).
Even though a device with more RAM and a faster processor could compute a barrier for a 3-qubit system, the same issue will arise when considering a 4, 5, or 6 qubit system.
It is difficult to see barrier certificates being used beyond a low number of qubits without the usage of a change in variables.
However, barrier certificates sacrifice scalability for flexibility as will be discussed in the conclusion.

\section{Conclusion}
\label{sec:conc}
In this article, we have considered a different approach to the verification of quantum circuits through the usage of barrier certificates by treating the quantum circuits as dynamical systems.
We firstly showed how to extend the notion of barrier certificates from the real domain into the complex domain.
We then demonstrated how the standard approach for generating barrier certificates, through sum of squares, can be extended to the complex domain as well, through Hermitian sum of squares.
Finally, we provided some experiments that show the usage of the developed technique.

\begin{figure}[t]
    \centering
    \includegraphics[width=0.8\textwidth]{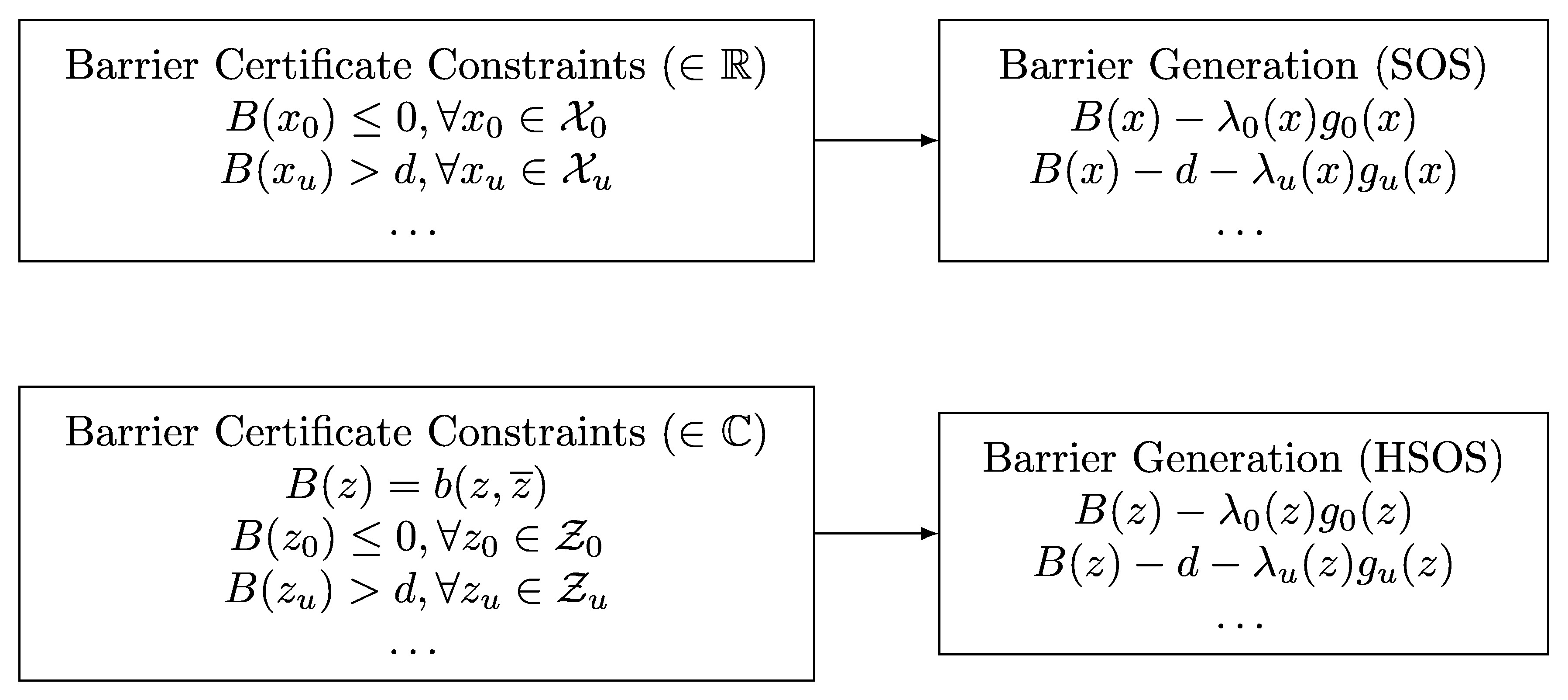}
    \caption{Conversion of a real barrier certificate scheme using SOS equations to generate a barrier into a complex scheme using HSOS equations.
    This is done by replacing the real variables with complex variables, allowing conjugation as an operation, setting the barrier function to be a conjugate-flattening polynomial and requiring the SOS equations to be HSOS equations instead.}
    \label{fig:sos2hsos}
    \Description{Some real barrier certificate constraints ($\in \mathbb{R}$) are generated using related SOS equations.
    Similarly, complex barrier certificate constraints ($\in \mathbb{C}$) are generated using related HSOS equations.}
\end{figure}

Our techniques for extending barrier certificates to the complex domain can be applied to complex dynamical systems in general.
Any barrier certificate that uses SOS techniques to generate a barrier can be extended to the complex domain through the use of conjugate-flattening functions and HSOS equations as demonstrated in Figure~\ref{fig:sos2hsos}.
This extension can be done almost freely by including the complex conjugate as an operation.

The barrier certificate technique is very flexible, but comes at the cost of scalability (as seen in Table~\ref{tab:results}).
While the barrier certificate approach is expensive for quantum circuits with a large number of qubits, the automation provided could be useful for verifying certain behaviors of quantum systems, which are continuous-time systems that are used to define the under-the-hood behavior of quantum gates.
We can model a quantum system when we have noisy qubit initialization and, in the future, we can consider how to verify quantum systems with additional properties, such as noise or control during state evolution, using barrier certificates.
Further, other behavioral properties can be investigated instead of safety, such as reachability where the system evolves into a specified region rather than avoid it.

\subsection*{Acknowledgements}
M. Lewis was supported by the UK EPSRC (project reference EP/T517914/1). P. Zuliani was supported by the project SERICS (PE00000014) under the Italian MUR National Recovery and Resilience Plan funded by the European Union - NextGenerationEU. The work of S. Soudjani is supported by the following grants: EPSRC EP/V043676/1, EIC 101070802, and ERC 101089047.

\bibliographystyle{ACM-Reference-Format}
\bibliography{refs}

\appendix
\section{$Z$ Gate Example Extended}
\label{app:case:z}
To extend the $Z$ gate example to an $n$-qubit system, the dynamics are set as
\begin{equation*}
    f_t(z) = Z^{\otimes n} z.
\end{equation*}
We then specify
\begin{equation*}
\begin{aligned}
    & \init{\sysspace}^p = \{ z \in \sysspace : \abs{(z)_p}^2 \geq 0.9 \}, \text{ and}\\
    & \unsafe{\sysspace}^p = \{ z \in \sysspace : \sum_{j\neq p} \abs{(z)_j}^2 \geq 0.2\},
\end{aligned}
\end{equation*}
for $p \in \{0, \dots, 2^n - 1\}$ (where $\ket{p}$ is the target state) for the initial and unsafe space respectively.

\end{document}